\documentclass[11pt]{article}

\usepackage{amsmath, amssymb, amsthm}
\usepackage{xspace}
\usepackage{comment}
\usepackage{mdwlist}
\usepackage{multicol}
\usepackage{wrapfig}
\usepackage{setspace}
\usepackage{url}
\usepackage[all]{xy}
\usepackage[hang,footnotesize,bf]{caption}
\usepackage[pdftex]{graphicx}

\newtheorem{theorem}{Theorem}[section]
\newtheorem{corollary}[theorem]{Corollary}
\newtheorem{lemma}[theorem]{Lemma}
\newtheorem*{definition}{Definition}
\newtheorem*{theorem*}{Theorem}
\newtheorem*{lemma*}{Lemma}

\long\def\symbolfootnote[#1]#2{\begingroup%
\def\thefootnote{\fnsymbol{footnote}}\footnote[#1]{#2}\endgroup}

\renewcommand{\thefootnote}{\fnsymbol{footnote}}

\newcommand{\VR}{\mathrm{VR}}

\newcommand{\Hom}{\mathrm{H_*}}
\newcommand{\Homi}{\mathrm{H_i}}

\newcommand{\ball}{\mathbf{ball}}

\newcommand{\dist}{\mathbf{d}}
\newcommand{\rdist}{\hat{\mathbf{d}}}

\newcommand{\R}{\mathbb{R}}
\newcommand{\Rnn}{\R_{\ge 0}}

\newcommand{\N}{\mathbb{N}}

\newcommand{\M}{\mathcal{M}}

\newcommand{\spread}{\Delta}
\newcommand{\e}{\varepsilon}

\newcommand{\id}{\mathrm{id}}

\newcommand{\rips}{\mathcal{R}}
\newcommand{\rrips}{\hat{\mathcal{R}}}
\newcommand{\srips}{\mathcal{Q}}
\newcommand{\ssrips}{\mathcal{S}}

\DeclareMathOperator*{\argmin}{argmin}

\newcommand{\dgk}{\mathrm{D}}

\newcommand{\rad}{\mathrm{rad}}
\newcommand{\parent}{\mathrm{par}}
\newcommand{\rep}{\mathrm{rep}}
\newcommand{\rel}{\mathrm{Rel}}
\newcommand{\child}{\mathrm{child}}
\newcommand{\net}{\mathcal{N}}

\newcommand{\cl}[1]{\overline{#1}}
\newcommand{\kpack}{K_{p}}

\begin{document}
  \title{Linear-Size Approximations to the Vietoris--Rips Filtration}
  \author{Donald R. Sheehy}
  \maketitle

  \begin{abstract}
  The Vietoris--Rips filtration is a versatile tool in topological data analysis.
  It is a sequence of simplicial complexes built on a metric space to add topological structure to an otherwise disconnected set of points.
  It is widely used because it encodes useful information about the topology of the underlying metric space.
  This information is often extracted from its so-called persistence diagram.
  Unfortunately, this filtration is often too large to construct in full.
  We show how to construct an $O(n)$-size filtered simplicial complex on an $n$-point metric space such that its persistence diagram is a good approximation to that of the Vietoris--Rips filtration.
  This new filtration can be constructed in $O(n\log n)$ time.
  The constant factors in both the size and the running time depend only on the doubling dimension of the metric space and the desired tightness of the approximation.
  For the first time, this makes it computationally tractable to approximate the persistence diagram of the Vietoris--Rips filtration across all scales for large data sets.
  
  We describe two different sparse filtrations.
  The first is a zigzag filtration that removes points as the scale increases.
  The second is a (non-zigzag) filtration that yields the same persistence diagram.
  Both methods are based on a hierarchical net-tree and yield the same guarantees.

\end{abstract}

  \section{Introduction} 
\label{sec:introduction}

  There is an extensive literature on the problem of computing sparse approximations to metric spaces (see the book~\cite{narasimhan07geometric} and references therein). 
  There is also a growing literature on topological data analysis and its efforts to extract topological information from metric data (see the survey~\cite{carlsson09topology} and references therein). 
  One might expect that topological data analysis would be a major user of metric approximation algorithms, especially given that topological data analysis often considers simplicial complexes that grow exponentially in the number of input points.
  Unfortunately, this is not the case.
  The benefits of a sparser representation are sorely needed, but it is not obvious how an approximation to the metric will affect the underlying topology.
  The goal of this paper is to bring together these two research areas and to show how to build sparse metric approximations that come with topological guarantees.

  The target for approximation is the Vietoris--Rips complex, which has a simplex for every subset of input points with diameter at most some parameter $\alpha$.
  The collection of Vietoris--Rips complexes at all scales yields the Vietoris--Rips filtration.
  The persistence algorithm takes this filtration and produces a persistence diagram representing the changes in topology corresponding to changes in scale~\cite{zomorodian05computing}.
  The Vietoris--Rips filtration has become a standard tool in topological data analysis because it encodes relevant and useful information about the topology of the underlying metric space~\cite{chazal09gromov-hausdorff}.
  It also extends easily to high dimensional data, general metric spaces, or even non-metric distance functions.
  
  Unfortunately, the Vietoris--Rips filtration has a major drawback: It's huge! 
  Even the $k$-skeleton (the simplices up to dimension $k$) has size $O(n^{k+1})$ for $n$ points.

  This paper proposes an alternative filtration called the sparse Vietoris--Rips filtration, which has size $O(n)$ and can be computed in $O(n\log n)$ time.
  Moreover, the persistence diagram of this new filtration is provably close to that of the Vietoris--Rips filtration.
  The constants depend only on the doubling dimension of the metric (defined below) and a user-defined parameter $\e$ governing the tightness of the approximation.
  For the $k$-skeleton, the constants are bounded by $\left(\frac{1}{\e}\right)^{O(kd)}$.
  
  The main tool we use to construct the sparse filtration is the net-tree of Har-Peled and Mendel~\cite{har-peled06fast}.
  Net-trees are closely related to hierarchical metric spanners~\cite{gao06deformable,gottlieb08optimal} and their construction is analogous to data structures used for nearest neighbor search in metric spaces~\cite{cole06searching, clarkson99nearest, clarkson03nearest}. 
  
  
  \paragraph*{Outline}
  After reviewing some related work and definitions in Sections~\ref{sec:related} and~\ref{sec:background}, 
  we explain how to perturb the input metric using weighted distances in Section~\ref{sec:relaxed_rips}.
  This perturbation is used in the definition of a sparse zigzag filtration in Section~\ref{sec:sparse_rips}, i.e. one in which simplices are both added and removed as the scale increases.
  The full definition of the net-trees is given in Section~\ref{sec:net_trees}.
  Using the properties of the net-tree and the perturbed distances, we prove in Section~\ref{sec:sparsification} that removing points from the filtration does not change the topology.
  This implies that the zigzag filtration does not actually zigzag at the homology level (Subsection~\ref{sec:reversing_zigzags}).
  The zigzag filtration can then be converted into an ordinary (i.e. non-zigzag) filtration that also approximates the Vietoris--Rips filtration (Subsection~\ref{sec:no_zigzag}).
  The theoretical guarantees are proven in Section~\ref{sec:theoretical_guarantees}.
  Subsection~\ref{sec:approximation_guarantee} proves that the resulting persistence diagrams are good approximations to the persistence diagram of the full Vietoris--Rips filtration.
  The size complexity of the sparse filtration is shown to be $O(n)$ in Subsection~\ref{sec:linear_complexity}.
  Finally, in Section~\ref{sec:construction}, we outline the $O(n\log n)$-time construction, which turns out to be quite easy once you have a net-tree.
  

  \section{Related Work} 
\label{sec:related}


  The theory of persistent homology~\cite{edelsbrunner02topological,zomorodian05computing} gives an algorithm for computing the persistent topological features of a complex that grows over time.
  It has been applied successfully to many problem domains, including image analysis~\cite{carlsson08local}, biology~\cite{singh08topological, chung09persistence}, and sensor networks~\cite{silva07homological,silva07coverage}.
  See also the survey by Carlsson for background on the topological view of data~\cite{carlsson09topology}.
  It is also possible to consider the complexes that alternate between growing and shrinking in what is known as zigzag persistence~\cite{tausz11applications,carlsson10zigzag,carlsson09zigzag,milosavljevic11zigzag}.

  
  Due to the rapid blowup in the size of the Vietoris--Rips filtration, some attempts have been made to build approximations.
  Some notable examples include witness complexes~\cite{boissonat07manifold,guibas07reconstruction,de-silva04topological} as well as the mesh-based methods of Hudson et al.\ in Euclidean spaces~\cite{hudson10topological}.

  The work most similiar to the current paper is by Chazal and Oudot~\cite{chazal08towards}.
  In that paper, they looked at a sequence of persistence diagrams on denser and denser subsamples.
  However, they were not able to combine these diagrams into a single diagram with a provable guarantee.
  Moreover, they were not able to prove general guarantees on the size of the filtration except under very strict assumptions on the data.
  
  Recently, Zomorodian~\cite{zomorodian10tidy} and Attali et al.~\cite{attali11efficient} have presented new methods for simplifying Vietoris--Rips complexes.
  These methods depend only on the combinatorial structure.  
  However, they have not yielded results in simplifying filtrations, only static complexes.
  In this paper, we exploit the geometry to get topologically equivalent sparsification of an entire filtration.
  
  \section{Background} 
\label{sec:background}

    \paragraph*{Doubling metrics} 

  For a point $p\in P$ and a set $S\subseteq P$, we will write $\dist(p, S)$ to denote the minimum distance from $p$ to $S$, i.e.\ $\dist(p,S) = \min_{q\in S}\dist(p,q)$.
  In a metric space $\M = (P,\dist)$, a \textbf{metric ball} centered at $p\in P$ with radius $r\in \R$ is the set $\ball(p,r) = \{q\in P : \dist(p,q)\le r\}$.
    
  \begin{definition}\label{def:doubling_dimension}
    The \textbf{doubling constant} $\lambda$ of a metric space $\M = (P,\dist)$ is the minimum number of metric balls of radius $r$ required to cover any ball of radius $2r$.
    The \textbf{doubling dimension} is $d = \lceil \lg \lambda \rceil$.
    A metric space whose doubling dimension is bounded by a constant is called a \textbf{doubling metric}.
  \end{definition}

  The \textbf{spread} $\spread$ of a metric space $\M = (P,\dist)$ is the ratio of the largest to smallest interpoint distances.
  A metric with doubling dimension $d$ and spread $\spread$ has at most $\spread^{O(d)}$ points.
  This is easily seen by starting with a ball of radius equal to the largest pairwise distance and covering it with $\lambda$ balls of half this radius.
  Covering all of the resulting balls by yet smaller balls and repeating $O(\log \spread)$ times results in balls that can contain at most one point each because the radii are smaller than the minimum interpoint distance.
  The number of such balls is $\lambda^{O(\log \spread)} = \spread^{O(\log \lambda)} = \spread^{O(d)}$.

    \paragraph*{Simplicial Complexes} 

    A \textbf{simplicial complex} $X$ is a collection of \textbf{vertices} denoted $V(X)$ and a collection of subsets of $V(X)$ called \textbf{simplices} that is closed under the subset operation, i.e.\ $\sigma\subset \psi$ and $\psi\in X$ together imply $\sigma\in X$.
    The \textbf{dimension} of a simplex $\sigma$ is $|\sigma| - 1$, where $|\cdot|$ denotes cardinality.
  %
  Note that this definition is combinatorial rather than geometric.
  These abstract simplicial complexes are not necessarily embedded in a geometric space.


    \paragraph*{Homology} 
  In this paper we will use simplicial homology over a field (see Munkres~\cite{munkres84elements} for an accessible introduction to algebraic topology).
  Thus, given a space $X$, the homology groups $\Homi(X)$ are vector spaces for each $i$.
  Let $\Hom(X)$ denote the collection of these homology groups for all $i$.

  The star subscript denotes the homomorphism of homology groups induced by a map between spaces, i.e. $f:X\to Y$ induces $f_{\star}:\Hom(X)\to\Hom(Y)$.
  We recall the functorial properties of the Homology operator, $\Hom(\cdot)$.
  In particular, $(f \circ g)_{\star} = f_\star \circ g_\star$ and $\id_{X\star} = \id_{\Hom(X)}$, where $\id$ indicates the identity map.


\paragraph*{Persistence Modules and Diagrams} 

A \textbf{filtration} is a nested sequence of topological spaces:
  $X_1 \subseteq X_2 \subseteq \cdots \subseteq X_n$.
  If the spaces are simplicial complexes (as with all the filtrations in this paper), then it is called a \textbf{filtered simplicial complex} (see Figure~\ref{fig:filtrations}, top).
  
  \begin{figure}[ht]
    \centering
      \includegraphics[width=.75\textwidth]{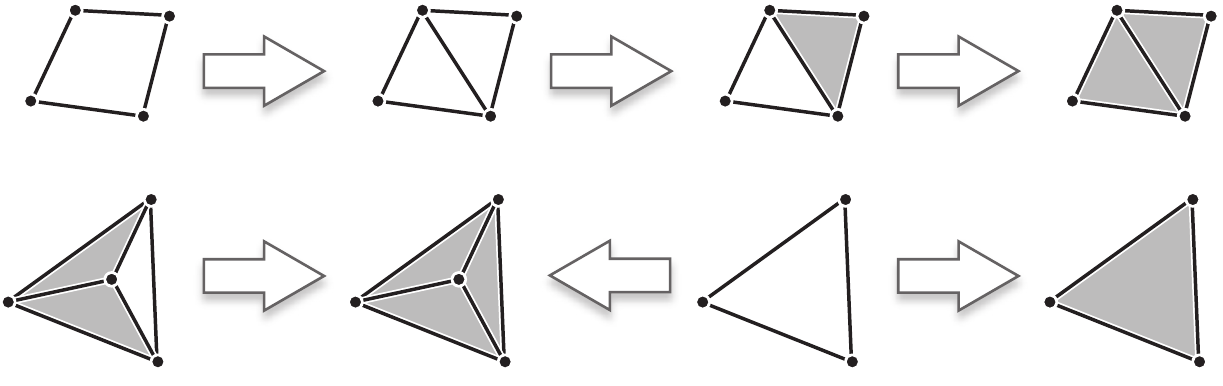}
    \caption{\label{fig:filtrations}
      Top: A filtered simplicial complex.  Bottom: A zigzag filtration of simplicial complexes.
    }
  \end{figure}

  A \textbf{persistence module} is a sequence of Homology groups connected by homomorphisms:
  \[
    \Hom(X_1) \to \Hom (X_2) \to \cdots \to \Hom(X_n).
  \]
  The homology functor turns a filtration with inclusion maps $X_1\hookrightarrow X_2 \hookrightarrow \cdots $ into a persistence module, but as we will see, this is not the only way to get one.

  One can also consider \textbf{zigzag filtrations}, which allow the inclusions to go in both directions:
  $X_1 \subseteq X_2 \supseteq X_3 \subseteq \cdots$.
  The resulting module is called a \textbf{zigzag module}.
  \[
    \Hom(X_1) \to \Hom (X_2) \leftarrow \Hom(X_3) \to \cdots.
  \]

  The \textbf{persistence diagram} of a persistence module is a multiset of points in $(\R \cup \{\infty\})^2$.
  Each point of the diagram represents a topological feature.
  The $x$ and $y$ coordinates of the points are the \emph{birth} and \emph{death} times of the feature and correspond to the indices in the persistence module where that feature appears and disappears.
  Points far from the diagonal persisted for a long time, while those ``non-persistent'' points near the diagonal may be considered \emph{topological noise}.
  By convention, the persistence diagram also contains every point $(x,x)$ of the diagonal with infinite multiplicity.
  
  Given a filtration $\mathcal{F}$, we let $\dgk\mathcal{F}$ denote the persistence diagram of the persistence module generated by $\mathcal{F}$.
  The persistence algorithm computes a persistence diagram from $\mathcal{F}$~\cite{zomorodian05computing}.
  It is also known how to compute a persistence diagram when $\mathcal{F}$ is a zigzag filtration~\cite{carlsson10zigzag,milosavljevic11zigzag}.


    \paragraph*{Approximating Persistence Diagrams} 

Given two filtrations $\mathcal{F}$ and $\mathcal{G}$, we say that the persistence diagram $\dgk\mathcal{F}$ is a \textbf{$c$-approximation to the diagram} $\dgk\mathcal{G}$ if there is a bijection $\phi:\dgk\mathcal{F}\to \dgk\mathcal{G}$ such that for each $p\in \dgk\mathcal{F}$, the birth times of $p$ and $\phi(p)$ differ by at most a factor of $c$ and the death times also differ by at most a factor of $c$. 
  The reader familiar with stability results for persistent homology~\cite{cohen-steiner07stability, chazal09proximity} will recognize this as bounding the $\ell_\infty$-bottleneck distance between the persistence diagrams after reparameterizing the filtrations on a $\log$-scale.

  We will make use of two standard results on persistence diagrams.
  The first gives a sufficient condition for two persistence modules to yield identical persistence diagrams.
  \begin{theorem}\label{thm:persistence_equivalence}[Persistence Equivalence Theorem~\cite[page 159]{edelsbrunner09computational}]
    Consider two sequences of vector spaces connected by homomorphisms $\phi_i:U_i\to V_i$:
    \[
      \xymatrix{
        V_0 \ar[r] & V_1 \ar[r] & \cdots \ar[r] & V_{n-1} \ar[r] & V_n \\
        U_0 \ar[r] \ar[u] & U_1 \ar[u] \ar[r] & \cdots \ar[r] & U_{n-1} \ar[r] \ar[u] & U_n \ar[u] \\
      }
    \]
    %
    If the vertical maps are isomorphisms and all squares commute then the persistence diagram defined by the $U_i$ is the same as that defined by the $V_i$. 
  \end{theorem}

  We prove approximation guarantees for persistence diagrams using the following lemma, which is a direct corollary of the Strong Stability Theorem of Chazal et al.~\cite{chazal09proximity} rephrased in the language of approximate persistence diagrams.
  \begin{lemma}\label{lem:persistence_approximation}[Persistence Approximation Lemma]
    For any two filtrations $\mathcal{A} = \{A_\alpha\}_{\alpha\ge 0}$ and $\mathcal{B}= \{B_\alpha\}_{\alpha\ge 0}$,
    if $A_{\alpha/c}\subseteq B_{\alpha}\subseteq A_{c\alpha}$ for all $\alpha\ge 0$ then
    the persistence diagram $\dgk \mathcal{A}$ is a $c$-approximation to the persistence diagram of $\dgk \mathcal{B}$. 
  \end{lemma}

    \paragraph*{Contiguous Simplicial Maps} 



  Contiguity gives a discrete version of homotopy theory for simplicial complexes.

  \begin{definition}\label{def:simplicial_map}
    Let $X$ and $Y$ be simplicial complexes.  
    A \textbf{simplicial map} $f:X\to Y$ is a function that maps vertices of $X$ to vertices of $Y$ 
    and $f(\sigma) := \bigcup_{v\in \sigma}f(v)$ is a simplex of $Y$ for all $\sigma\in X$.
  \end{definition}

  A simplicial map is determined by its behavior on the vertex set.
  Consequently, we will abuse notation slightly and identify maps between vertex sets and maps between simplices.
  When it is relevant and non-obvious, we will always prove that the resulting map between simplicial complexes is simplicial.

  \begin{definition}\label{def:contiguous}
    Two simplicial maps $f,g:X\to Y$ are \textbf{contiguous} if $f(\sigma)\cup g(\sigma)\in Y$ for all $\sigma\in X$.
  \end{definition}
  
  
  \begin{definition}\label{def:retraction}
    For any pair of topological spaces $X\subseteq Y$, a map $f:Y\to X$ is a \textbf{retraction} if $f(x)= x$ for all $x\in X$. 
    Equivalently, $f \circ i = \id_X$ where $i:X\hookrightarrow Y$ is the inclusion map.
  \end{definition}

  The theory of contiguity is a simplicial analogue of homotopy theory.
  If two simplicial maps are contiguous then they induce identical homomorphisms at the homology level~\cite[\S 12]{munkres84elements}.
  The following lemma gives a homology analogue of a deformation retraction.

  \begin{lemma}\label{lem:induced_isomorphism}
    Let $X$ and $Y$ be simplicial complexes such that $X\subseteq Y$ and let $i:X\hookrightarrow Y$ be the canonical inclusion map.
    If there exists a simplicial retraction $\pi:Y\to X$ 
    such that $i\circ \pi$ and $\id_Y$ are contiguous, 
    then $i$ induces an isomorphism $i_\star:\Hom(X)\to \Hom(Y)$ between the corresponding homology groups.
  \end{lemma}
  \begin{proof}
    Since $i \circ \pi$ and $\id_Y$ are contiguous, the induced homomorphisms $(i \circ \pi)_\star:\Hom(Y)\to \Hom(Y)$ and $\id_{Y\star}:\Hom(Y)\to \Hom(Y)$ are identical~\cite[\S 12]{munkres84elements}.
    Since $\id_{Y\star}=(i \circ \pi)_\star = i_\star \circ \pi_\star$ is an isomorphism, it follows that $i_\star$ is surjective.

    Since $\pi$ is a retraction, $\pi\circ i = \id_X$ and thus 
    $(\pi \circ i)_\star:\Hom(X)\to \Hom(X)$ and $\id_{X\star}:\Hom(X)\to \Hom(X)$ are identical.
    Since $\id_{X\star}=(\pi \circ i)_\star = \pi_\star \circ i_\star$ is an isomorphism, it follows that $i_\star$ is injective.
    
    Thus, $i_\star$ is an isomorphism because it is both injective and surjective.
  \end{proof}

    
  \section{The Relaxed Vietoris--Rips Filtration} 
\label{sec:relaxed_rips}

  In this section, we relax the input metric so that it is no longer a metric, but it will still be provably close to the input.
  The new distance adds a small weight to each point that grows with $\alpha$.
  The intuition behind this process is illustrated in Figure~\ref{fig:points_on_a_line}.
  The weighted distance effectively shrinks the metric balls locally so that one ball may be covered by nearby balls.
  
  \begin{figure}[ht]
    \centering
      \includegraphics[width=.75\textwidth]{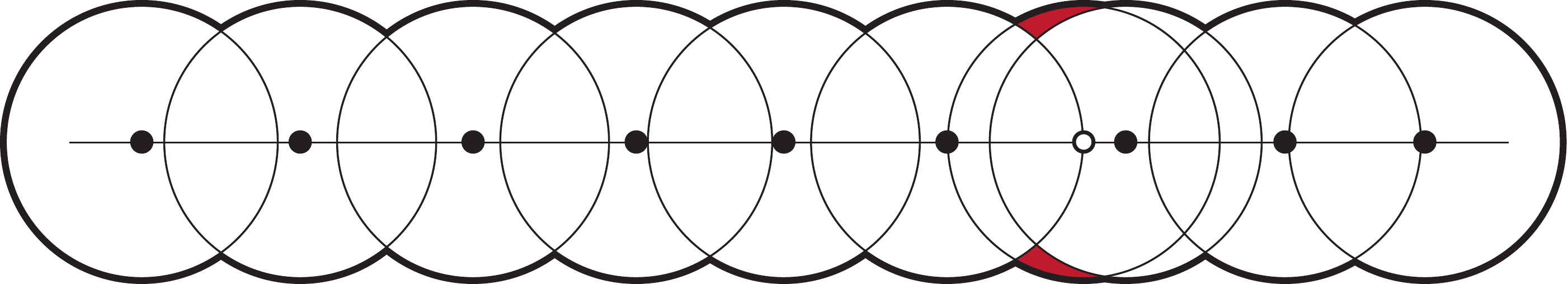}
      \includegraphics[width=.75\textwidth]{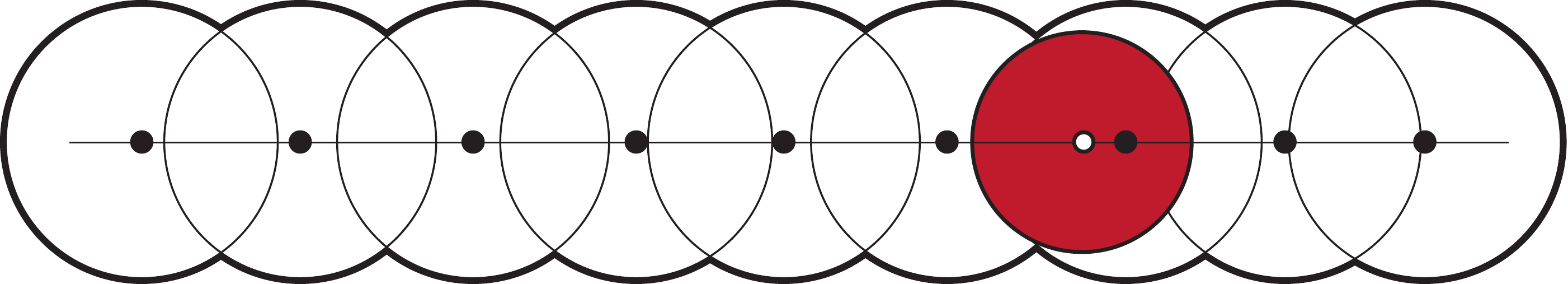}
    \caption{\label{fig:points_on_a_line}
      Top: Some points on a line.  The white point contributes little to the union of $\alpha$-balls.  Bottom: Using the relaxed distance, the new $\alpha$-ball is completely contained in the union of the other balls.  Later, we use this property to prove that removing the white point will not change the topology.
    }
  \end{figure}
  
  Throughout, we assume the user-defined parameter $\e\le \frac{1}{3}$ is fixed.
  Each point $p$ is assigned a \textbf{deletion time} $t_p\in \Rnn$.
  The specific choice of $t_p$ will come from the net-tree construction in Section~\ref{sec:net_trees}.
  For now, we will assume the deletion times are given, assuming only that they are nonnegative.
  The weight $w_p(\alpha)$ of point $p$ at scale $\alpha$ is defined as
  \[
    w_p(\alpha) := \left\{
      \begin{array}{ll}
        0 & \text{if $\alpha \le (1-2\e)t_p$}\\
        \frac{1}{2}(\alpha - (1-2\e)t_p) & \text{if $(1-2\e)t_p < \alpha < t_p $} \\
        \e \alpha & \text{if $t_p \le \alpha $}
      \end{array}
    \right.
  \]
  \begin{figure}[ht]
    \centering
      \includegraphics[width=.75\textwidth]{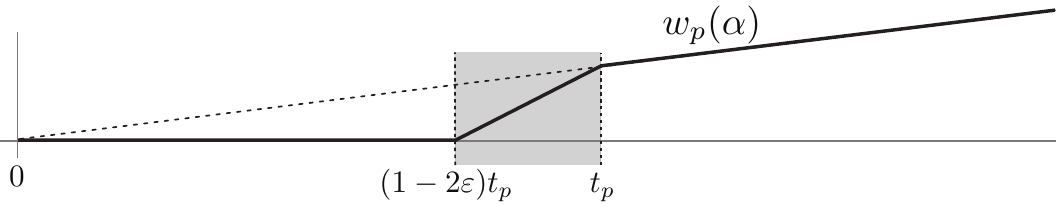}
    \caption{\label{fig:weights}
      The weight function for a point $p$.
      The weight is $0$ until just before its removal time $t_p$.
      Then there is a period of steeper increase (slope $= 1/2$) followed by slower increase (slope $= \e$).
    }
  \end{figure}

  The \textbf{relaxed distance} at scale $\alpha$ is defined as
  \[
    \rdist_\alpha(p,q) := \dist(p, q) + w_p(\alpha) + w_q(\alpha).
  \]
  For any pair $p,q\in P$, the relaxed distance $\rdist_\alpha(p, q)$ is monotonically non-decreasing in $\alpha$.
  In particular, $\rdist_\alpha\ge \rdist_0 = \dist$ for all $\alpha\ge 0$.
  Although distances can grow as $\alpha$ grows, this growth is sufficiently slow to allow the following lemma which will be useful later.
  \\
  
  \begin{lemma}\label{lem:rdist_induces_filtrations}
    If $\rdist_\alpha(p,q)\le \alpha \le \beta$ then $\rdist_\beta(p,q)\le \beta$.
  \end{lemma}
  \begin{proof}
    The weight of a point is $\frac{1}{2}$-Lipschitz in $\alpha$, so $w_p(\beta)\le w_p(\alpha) + \frac{1}{2}|\beta-\alpha|$, and similarly, $w_q(\beta)\le w_q(\alpha) + \frac{1}{2}|\beta-\alpha|$.
    So,
    \begin{align*}    
      \rdist_\beta(p,q) 
        &= \dist(p,q) + w_p(\beta) + w_q(\beta)\\
        &\le \dist(p,q) + w_p(\alpha) + w_q(\alpha) + (\beta-\alpha)\\
        &= \rdist_\alpha(p,q) + \beta - \alpha \\
        &\le \beta
    \end{align*}
  \end{proof}
  
  Given a set $P$, a distance function $\dist':P\times P\to \R$, and a scale parameter $\alpha\in \R$, we can construct a \textbf{Vietoris--Rips complex}
  \[
    \VR(P,\dist', \alpha) := \{ \sigma\subset P : \dist'(p,q)\le \alpha \text{ for all $p,q\in \sigma$}\}.
  \]
  The Vietoris--Rips complex associated with the input metric space $(P,\dist)$ is 
  $\rips_\alpha := \VR(P,\dist,\alpha)$.
  The \textbf{relaxed Vietoris--Rips complex} is
  $\rrips_\alpha := \VR(P,\rdist_\alpha,\alpha)$.

  By considering the family of Vietoris--Rips complexes for all values of $\alpha\ge 0$, we get the \textbf{Vietoris--Rips filtration}, $\rips:=\{\rips_\alpha\}_{\alpha\ge 0}$.  
  Similarly, we may define the \textbf{relaxed Vietoris--Rips filtration}, $\rrips:=\{\rrips_\alpha\}_{\alpha\ge 0}$.
  Lemma~\ref{lem:rdist_induces_filtrations} implies that $\rrips$ is indeed a filtration.
  The filtrations $\rips$ and $\rrips$ are very similar.
  The following lemma makes this similarity precise via a multiplicative interleaving.

  \begin{lemma}\label{lem:rips_relaxed_rips_interleave}
    For all $\alpha\ge 0$, $\rips_{\frac{\alpha}{c}} \subseteq \rrips_\alpha \subseteq \rips_\alpha$, where $c = \frac{1}{1-2\e}$.
  \end{lemma}
  \begin{proof}
    To prove inclusions between Vietoris--Rips complexes, it suffices to prove inclusion of the edge sets.
    For the first inclusion, we must prove that for any pair $p,q$, if $\dist(p,q)\le \frac{\alpha}{c}$ then $\rdist_\alpha(p,q)\le \alpha$.
    Fix any such pair $p,q$.
    By definition, $w_p(\alpha)\le \e \alpha$ and $w_q(\alpha)\le \e \alpha$.
    So,
    \[  
      \rdist_{\alpha}(p,q) 
        = \dist(p,q) + w_p(\alpha) + w_q(\alpha)  
        \le \frac{\alpha}{c} + 2 \e \alpha  
        = \alpha. 
    \]

    For the second inclusion, $\rdist_\alpha\ge \dist$.
    So, if $\rdist_\alpha(p,q)\le \alpha$ then $\dist(p,q)\le \alpha$ as well.
    Thus any edge of $\rrips_\alpha$ is also an edge of $\rips_\alpha$.
  \end{proof}


  \section{The Sparse Zigzag Vietoris--Rips Filtration} 
\label{sec:sparse_rips}

  We will construct a sparse subcomplex of the relaxed Vietoris--Rips complex $\rrips_\alpha$ that is guaranteed to have linear size for any $\alpha$.
  In fact, we will get a zigzag filtration that only has a linear total number of simplices, yet its persistence diagram is identical to that of the relaxed Vietoris--Rips filtration.

  We define the \textbf{open net} $\net_\alpha$ at scale $\alpha$ to be the subset of $P$ with deletion time greater than $\alpha$, i.e.\ 
  \[
    \net_\alpha:=\{p\in P : t_p > \alpha\}.
  \]  
  Similarly, the \textbf{closed net} at scale $\alpha$ is 
  \[
    \cl\net_\alpha:=\{p\in P : t_p \ge \alpha\}.
  \]


  The \textbf{sparse zigzag Vietoris--Rips complex} $\srips_\alpha$ at scale $\alpha$ is just the subcomplex of $\rrips_\alpha$ induced on the vertices of $\net_\alpha$.
  Formally,
  \[
    \srips_\alpha := \{\sigma\in \rrips_\alpha : \sigma\subseteq \net_\alpha\} = \VR(\net_\alpha, \rdist_\alpha, \alpha).
  \]
  We also define a closed version of the sparse zigzag Vietoris--Rips complex:
  \[
    \cl\srips_\alpha := 
    \VR(\cl\net_\alpha, \rdist_\alpha, \alpha).
  \]  
  Note that if $\alpha\neq t_p$ for all $p\in P$ then $\net_\alpha = \cl\net_\alpha$ and $\srips_\alpha = \cl\srips_\alpha$.
  
  The complexes $\rips_\alpha$, $\rrips_\alpha$, $\srips_\alpha$, and $\cl\srips_\alpha$ are well-defined for all $\alpha\ge 0$, however, they only change at discrete scales.
  Let $A = \{a_i\}_{i\in\N}$ be an ordered, discrete set of nonnegative real numbers such that $t_p\in A$ for all $p\in P$ and $\alpha\in A$ for any pair $p,q$ such that $\alpha = \rdist_\alpha(p,q)$.
  That is, $A$ contains every scale at which a combinatorial changes happens, either a point deletion or an edge insertion.
  This implies that $\net_{a_i} = \cl\net_{a_{i+1}}$ and thus, using Lemma~\ref{lem:rdist_induces_filtrations}, that $\srips_{a_i}\subseteq \cl\srips_{a_{i+1}}$.

  The sparse Vietoris--Rips complexes can be arranged into a zigzag filtration $\srips$ as follows.
  \[
    \cl\srips_{a_1} \hookleftarrow
    \srips_{a_1} \hookrightarrow
    \cl\srips_{a_2} \hookleftarrow
    \srips_{a_2} \hookrightarrow
    \cdots
  \]

  We will return to $\srips$ later as it has some interesting properties.
  However, at this point, it is underspecified as we have not yet shown how to compute the deletion times for the vertices.
  The next section will fill this gap.

  \section{Hierarchical Net-Trees} 
\label{sec:net_trees}

  The following treatment of net-trees is adapted from the paper by Har-Peled and Mendel~\cite{har-peled06fast}.

  \begin{definition}\label{def:net-tree}
    A \textbf{net-tree} of a metric $\M =(P,\dist)$ is a rooted tree $T$ with vertex $v\in T$ having a \textbf{representative point} $\rep(v)\in P$.
    There are $n = |P|$ leaves, each represented by a different point of $P$.
    Each non-root vertex $v\in T$ has a unique \textbf{parent} $\parent(v)$.
    The set of vertices with the same parent $v$ are called the \textbf{children} of $v$, denoted $\child(v)$.
    If $\child(v)$ is nonempty then for some $u\in \child(v)$, $\rep(u) = \rep(v)$.
    The set $P_v\subseteq P$ denotes the points represented by the leaves of the subtree rooted at $v$.
    Each vertex $v\in T$ has an associated radius $\rad(v)$ satisfying the following two conditions.
    \begin{enumerate}
      \item \textbf{Covering Condition:} $P_v\subset \ball(\rep(v),\rad(v))$, and
      \item \textbf{Packing Condition:} if $v$ is not the root, then 
      \[
        P\cap \ball(\rep(v), \kpack \rad(\parent(v)))\subseteq P_v,
      \]
      where $\kpack$ (the ``$p$'' is for ``packing'') is a constant independent of $\M$ and $n$.
    \end{enumerate}
  \end{definition}

  \begin{figure}[ht]
    \centering
      \includegraphics[width=\textwidth]{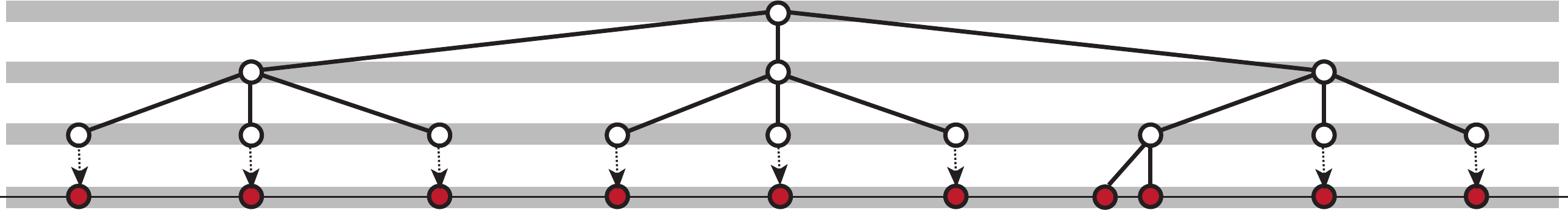}
    \caption{\label{fig:net_tree}
      A net-tree is built over the set points from Figure~\ref{fig:points_on_a_line}.
      Each level of the tree represents a sparse approximation to the original point set at a different scale.
    }
  \end{figure}

  The radii of the net tree nodes are always some constant times larger than the radius of their children.
  Simple packing arguments guarantee that no node of the tree has more than $\lambda^{O(1)}$ children, where $\lambda$ is the doubling constant of the metric.
  The whole tree can be constructed in $O(n\log n)$ randomized time or in $O(n\log \Delta)$ time deterministically~\cite{har-peled06fast}.
  Moreover, it is important to note the construction does not require that we know the doubling dimension in advance.
  
  Given a net-tree $T$ for $\M=(P,\dist)$ and a point $p\in P$, let $v_p$ denote the least ancestor among the nodes in $T$ represented by $p$.
  For each $p\in P$ the deletion time $t_p$ is defined as
  \[
    t_p := \frac{1}{\e(1-2\e)}\rad(\parent(v_p)).
  \]
  This is just the radius of the parent of $v_p$ with a small scaling factor included for technical reasons.
  When the scale $\alpha$ reaches $t_p$, we remove point $p$ from the (zigzag) filtration.
  The choice of weights as a function of $t_p$ guarantees that any point with relaxed distance at most $t_p$ from $p$ will also have relaxed distance at most $t_p$ from $\rep(\parent(v_p))$.
  As we prove later, this guarantees that the topology of the Rips complex does not change when we remove $p$ at scale $t_p$.

  \begin{figure}[ht]
    \centering
      \includegraphics[width=0.31\textwidth]{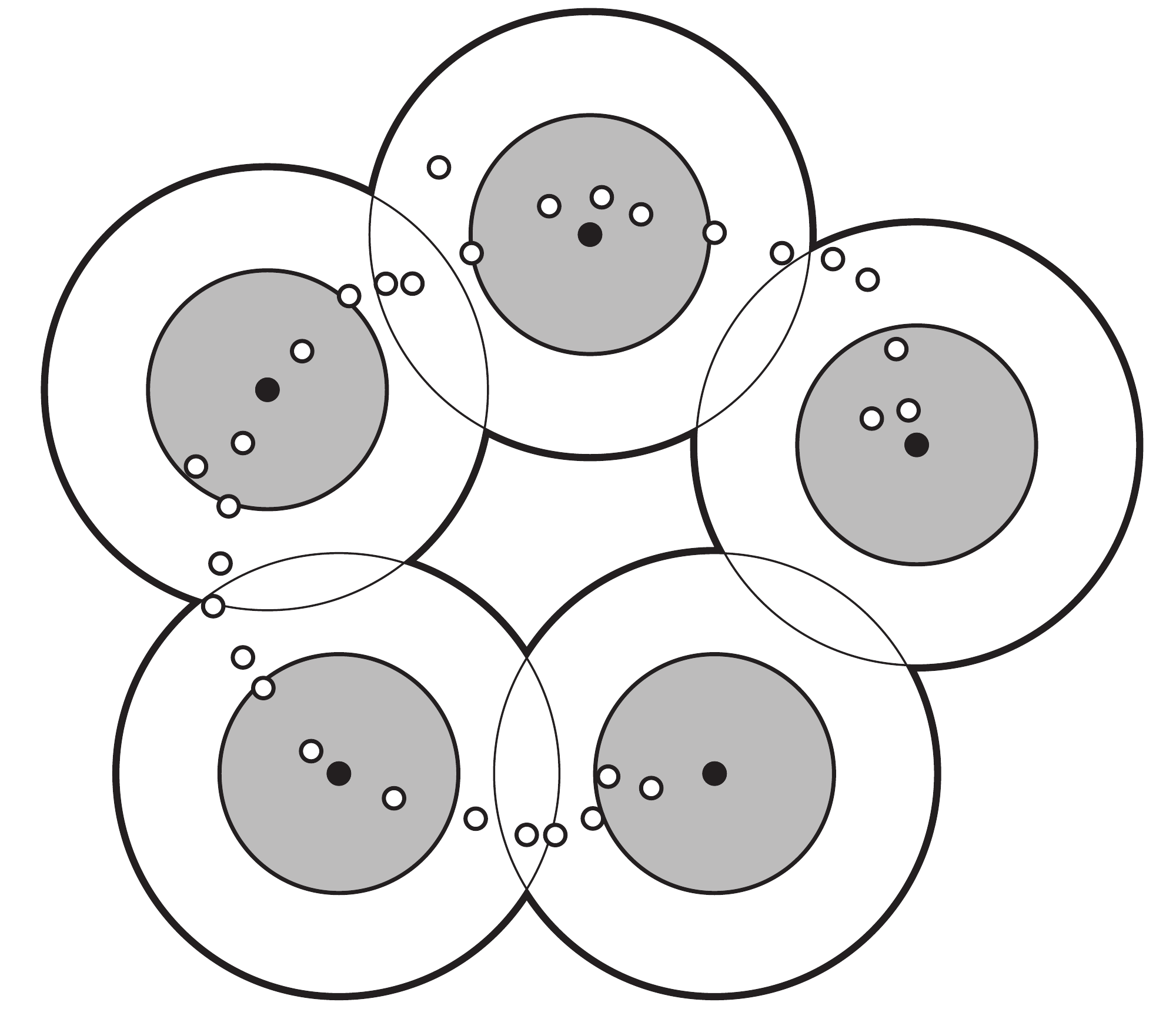}
      \includegraphics[width=0.31\textwidth]{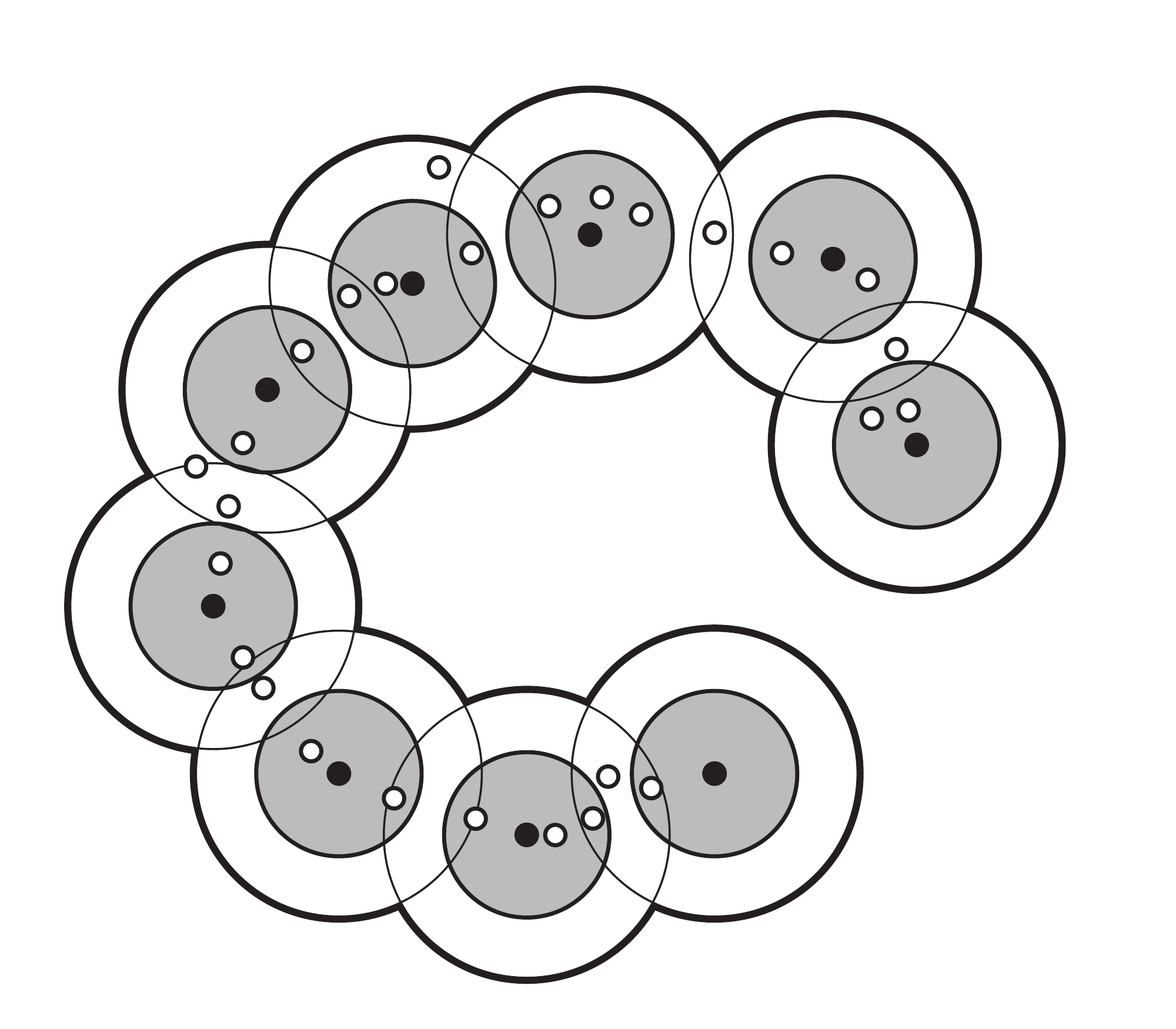}
      \includegraphics[width=0.31\textwidth]{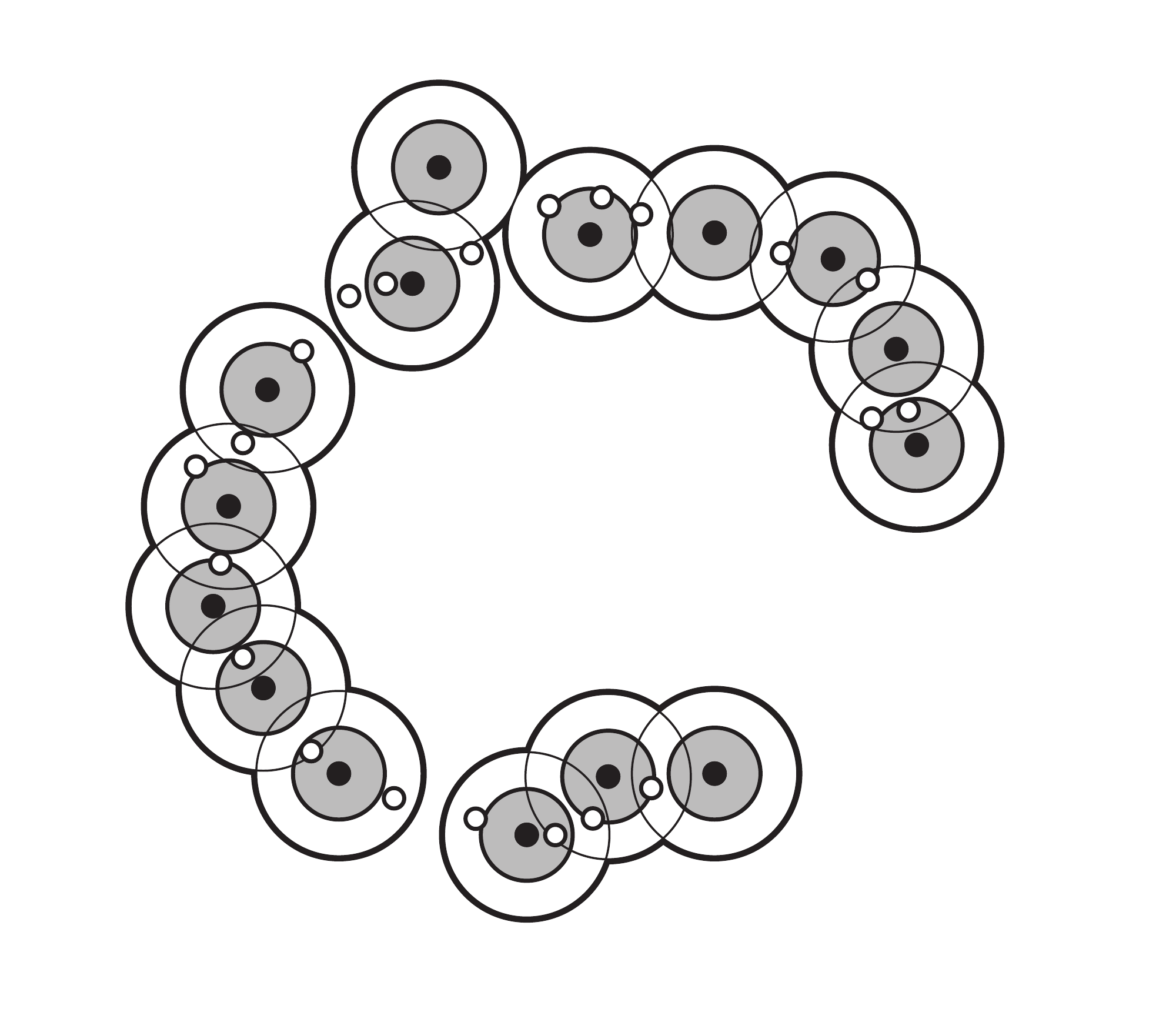}
    \caption{\label{fig:pack_cover}
      Three levels of a net tree for a collection of points in the plane.  In each level, the larger disks cover the rest of the points, whereas the smaller disks are disjoint, i.e.\ they pack. 
    }
  \end{figure}

  For a fixed scale $\alpha\in \R$, the set $\net_\alpha$
  is a subset of points of $P$ induced by the net-tree.
  The sets $\net_\alpha$ are the \textbf{nets} of the net-tree.
  For any $\alpha$ it satisfies a packing condition and a covering condition as defined in the following lemma.
  
  \begin{lemma}\label{lem:nets_pack_and_cover}
    Let $\M = (P,\dist)$ be a metric space and let $T$ be a net-tree for $\M$.  
    For all $\alpha\ge 0$, the net $\net_\alpha$ induced by $T$ at scale $\alpha$ satisfies the following two conditions.
    \begin{enumerate}
      \item \textbf{Covering Condition:} For all $p\in P$, $\dist(p,\net_\alpha) \le \e(1-2\e)\alpha$.
      \item \textbf{Packing Condition:} For all distinct $p,q\in \net_\alpha$, $\dist(p,q)\ge \kpack\e(1-2\e)\alpha$
    \end{enumerate}
  \end{lemma}
  \begin{proof}
    First, we prove that the covering condition holds.
    Fix any $p\in P$.
    The statement is trivial if $p\in \net_\alpha$ so we may assume that $t_p\le\alpha$.
    
    Let $v$ be the lowest ancestor of $p$ in $T$ such that $\rep(v)\in \net_\alpha$.
    Let $u$ be the ancestor of $p$ among the children of $v$.
    If $q = \rep(u)$ then $t_q = \frac{\rad(v)}{\e(1-2\e)}$.
    By our choice of $v$, $q\notin \net_\alpha$ and thus $t_q\le\alpha$.
    It follows that $\rad(v)\le \e(1-2\e)\alpha$.
%
%
    Thus, $\dist(p,\net_\alpha)\le \dist(p, \rep(v))\le \rad(v) \le \e(1-2\e)\alpha$.
    
    We now prove that the packing condition holds.
    Let $p,q$ be any two distinct points of $\net_\alpha$.
    Without loss of generality, assume $t_p\le t_q$.
    Thus, $q\notin P_{v_p}$, where (as before) $v_p$ is the least ancestor among the nodes of $T$ represented by $p$. 
    Since $p\in \net_\alpha$, $\alpha < t_p = \frac{\rad(\parent(v_p))}{\e(1-2\e)}$.
    Therefore, using the packing condition on the net-tree $T$, 
    $\dist(p,q)\ge \kpack \rad(\parent(v_p)) > \kpack \e(1-2\e) \alpha$.
  \end{proof}

  A subset that satisfies this type of packing and covering conditions is sometimes referred to as a metric space net (not to be confused with a range space net) or, more accurately, as a Delone set~\cite{clarkson06building}.
  An example is given in Figure~\ref{fig:pack_cover}.

  \section{Topology-preserving sparsification} 
\label{sec:sparsification}

  In this section, we make the intuition of Figure~\ref{fig:points_on_a_line} concrete by showing that deleting a vertex $p$ (and its incident simplices) from the relaxed Vietoris--Rips complex $\rrips_{t_p}$ does not change the topology.

  For any $\alpha\ge 0$, we define the ``projection'' of $P$ onto $\net_\alpha$ as
  \begin{equation*}
    \pi_\alpha(p) := 
    \left\{
      \begin{array}{ll}
        p & \text{if $p\in \net_\alpha$}\\
        \displaystyle{\argmin_{q\in \net_\alpha} \rdist_\alpha(p,q)} & \text{otherwise.}
      \end{array}
    \right.
  \end{equation*}
  The following lemma shows that the distance from a point to its projection is bounded by the difference in the weights of the point and its projection.
  
  \begin{lemma}\label{lem:dist_point_to_projection}
    For all $p\in P$, $\dist(p, \pi_\alpha(p)) \le w_p(\alpha) - w_{\pi_\alpha(p)}(\alpha)$.
  \end{lemma}
  \begin{proof}
    Fix any $p\in P$.
    We first prove that if $\dist(p, q) \le w_p(\alpha) - w_q(\alpha)$ for some $q\in \net_\alpha$, then it holds for $q = \pi_\alpha(p)$.
    If we have such a $q$, then the definitions of $\rdist_\alpha$ and $\pi_\alpha$ imply the following.
    \begin{align*}      
      \dist(p, \pi_\alpha(p)) &=   \rdist_\alpha(p,\pi_\alpha(p)) - w_p(\alpha) - w_{\pi_\alpha(p)}(\alpha)\\ 
        &\le \rdist_\alpha(p,q)  - w_p(\alpha) - w_{\pi_\alpha(p)}(\alpha) \\
        &\le   \dist(p,q) + w_q(\alpha) - w_{\pi_\alpha(p)}(\alpha) \\
        &\le w_p(\alpha) - w_{\pi_\alpha(p)}(\alpha). 
    \end{align*}
    So, it will suffice to find a $q\in \net_\alpha$ such that $\dist(p,q)\le w_p(\alpha)-w_q(\alpha)$.
    If $p\in \net_\alpha$ then this is trivial.
    So we may assume $p\notin \net_\alpha$ and therefore $t_p\le\alpha$ and $w_p(\alpha) = \e\alpha$.
    
    Let $u\in T$ be the ancestor of $p$ such that $\rad(u) < \frac{\e \alpha}{1-\e}$ and $\rad(\parent(u))\ge \frac{\e \alpha}{1-\e}$.
    Let $q = \rep(u)$.
    Since 
    \[
      t_q\ge\frac{1}{\e(1-2\e)}\rad(\parent(u)) \ge \frac{\alpha}{1-2\e}, 
    \]
    it follows that $w_q(\alpha) = 0$ and that $q\in \net_\alpha$.
    Finally, since $p\in P_u$, $\dist(p,q) \le \rad(u) \le \e\alpha = w_p(\alpha) - w_q(\alpha)$.
  \end{proof}

  By bounding the distance between points and their projections, we can now show that distances in the projection do not grow.
  
  \begin{lemma}\label{lem:distances_dont_grow}
    For all $p,q\in P$ and all $\alpha\ge 0$, $\rdist_\alpha(\pi_\alpha(p),q)\le \rdist_\alpha(p,q)$.
  \end{lemma}
  \begin{proof}
    The bound follows from the definition of $\rdist_\alpha$, the triangle inequality, and Lemma~\ref{lem:dist_point_to_projection}.
  \end{proof}

  \begin{lemma}\label{lem:homology_isomorphic_sparsification}
    Let $\alpha\ge 0$ be a fixed constant.
    Let $X$ be a set of points such that $\net_\alpha\subseteq X \subseteq P$ and let $K = \VR(X, \rdist_\alpha, \alpha)$.  
    The inclusion map $i:\srips_\alpha \hookrightarrow K$ induces an isomorphism at the homology level.
  \end{lemma}
  \begin{proof}
    The map $K\to \srips_\alpha$ induced by $\pi_\alpha$ is a retraction because $\srips_\alpha\subseteq K$ and $\pi_\alpha$ is a retraction onto $\net_\alpha$, the vertex set of $\srips_\alpha$.
    By Lemma~\ref{lem:induced_isomorphism}, it will suffice to prove that $\pi_\alpha$ is simplicial and that $i \circ \pi_\alpha$ is contiguous to the identity map on $K$.
    Since $\srips_\alpha$ and $K$ are Vietoris--Rips complexes, it will suffice to prove these facts for the edges:
    \begin{enumerate*}
      \item $\pi_\alpha$ is simplicial: for all $p,q\in X$, if $\rdist_\alpha(p,q)\le \alpha$ then $\rdist_\alpha(\pi_\alpha(p), \pi_\alpha(q)) \le \alpha$, and  
      \item $i \circ \pi_\alpha$ and $\id_K$ are contiguous: for all $p,q\in X$, if $\rdist_\alpha(p,q)\le \alpha$ then all six edges of the tetrahedron $\{p,q,\pi_\alpha(p),\pi_\alpha(q)\}$ are in $K$.
    \end{enumerate*}
    The first statement follows from two successive applications of Lemma~\ref{lem:distances_dont_grow}.
    The second statement follows from Lemma~\ref{lem:dist_point_to_projection} for the edges $\{p,\pi_\alpha(p)\}$ and $\{q,\pi_\alpha(q)\}$ 
    and from Lemma~\ref{lem:distances_dont_grow} for the other edges.
  \end{proof}

  
  \begin{corollary}\label{cor:back_arrows_are_isomorphisms}
    For all $\alpha\in A$, the inclusions $f:\srips_{\alpha} \hookrightarrow \cl\srips_{\alpha}$, $g:\srips_{\alpha} \hookrightarrow \rrips_{\alpha}$, and $h:\cl\srips_{\alpha} \hookrightarrow \rrips_{\alpha}$  induce isomorphisms at the homology level.
  \end{corollary}
  \begin{proof}
    The inclusions $f$ and $g$ induce isomorphisms by applying Lemma~\ref{lem:homology_isomorphic_sparsification} with $X = \cl\net_\alpha$ and $X = P$ respectively.
    Composing the inclusions, we get that $g = h\circ f$.
    Thus, at the homology level, we get $h_\star = g_\star \circ f_\star^{-1}$ is also an isomorphism.
  \end{proof}

  
  \section{Straightening out the Zigzags}\label{sec:straightening_zigzags}
  In this section, we show two different ways in which zigzag persistence may be avoided.
  First, in Subsection~\ref{sec:reversing_zigzags}, we show that the sparse zigzag filtration $\srips$ does not zigzag at the homology level.
  Then, in Subsection~\ref{sec:no_zigzag}, we show how to modify the zigzag filtration so it does not zigzag as a filtration either.

  The advantage of the non-zigzagging filtration is that it allows one to use the standard persistence algorithm,
  but it has larger size in the intermediate complexes.
  As we will see in Subsection~\ref{sec:linear_complexity}, the total size is still linear.

    \subsection{Reversing Homology Isomorphisms} 
\label{sec:reversing_zigzags}

  The backwards arrows in the zigzag filtration $\srips$ all induce isomorphisms.
  At the homology level, these isomorphisms can be replaced by their inverses to give a persistence module that does not zigzag.
  That is, the zigzag module
  \[
    \cdots \to
    \Hom(\cl\srips_\alpha) \overset{\cong}{\leftarrow}
    \Hom(\srips_\alpha) \to
    \Hom(\cl\srips_\beta) \overset{\cong}{\leftarrow}
    \Hom(\srips_\beta) \to \cdots
  \]
  can be transformed into
  \[
    \cdots \to
    \Hom(\cl\srips_\alpha) \overset{\cong}{\to}
    \Hom(\srips_\alpha) \to
    \Hom(\cl\srips_\beta) \overset{\cong}{\to}
    \Hom(\srips_\beta) \to \cdots.
  \]
  The latter module implies the existence of another that only uses the closed sparse Vietoris--Rips complexes:
  \[
    \cdots \to
    \Hom(\cl\srips_\alpha) \to
    \Hom(\cl\srips_\beta) \to \cdots.
  \]
  Note that this module does not duplicate the indices in the zigzag.
  In these various transformations, we have only reversed or concatenated isomorphisms, thus we have not changed the rank of any induced map $\Hom(\cl\srips_\alpha) \to \Hom(\cl\srips_\beta)$.
  As a result the persistence diagram $\dgk\srips$ is unchanged.

  This is novel in that we construct a zigzag filtration and we apply the zigzag persistence algorithm, but we are really computing the diagram of a persistence module that does not zigzag.
  The zigzagging can then be interpreted as sparsifying the complex without changing the topology.

    \subsection{A Sparse Filtration without the Zigzag} 
\label{sec:no_zigzag}

  The preceding subsection showed that the sparse zigzag Vietoris--Rips filtration does not zigzag as a persistence module.
  This hints that it is possible to construct a filtration that does not zigzag with the same persistence diagram.
  Indeed, this is possible using the filtration 
  \[
    \ssrips := \{\ssrips_{a_k}\}_{a_k\in A},
    \text{~~where~~}
    \ssrips_{a_k} := \bigcup_{i = 1}^k\cl\srips_{a_i}.
  \]

  We first prove that $\Hom(\cl\srips_{a_k})$ and $\Hom(\ssrips_{a_k})$ are isomorphic.

  \begin{lemma}\label{lem:srips_in_ssrips_isomorphism}
    For all $a_k \in A$, the inclusion $h:\cl\srips_{a_k}\hookrightarrow \ssrips_{a_k}$ induces a homology isomorphism.
  \end{lemma}
  \begin{proof}
    We define some intermediate complexes that interpolate between $\cl\srips_{a_k}$ and $\ssrips_{a_k}$.
    \[
      T_{i,k} := \bigcup_{j=i}^k \cl\srips_{a_j}.
    \]
    In particular, we have that
      $\cl\srips_{a_k} = T_{k,k} \text{ and } \ssrips_{a_k} = T_{1,k}$.
    The map $h$ can be expressed as 
    $h = h_1\circ\cdots\circ h_{k-1}$, where
    $h_i:T_{i+1,k}\hookrightarrow T_{i,k}$ is an inclusion.
    It will suffice to prove that each $h_i$ induces an isomorphism at the homology level for each $i=1\ldots k-1$.
    By Lemma~\ref{lem:induced_isomorphism}, it will suffice to show that the projection $\pi_{a_i}:T_{i,k}\to T_{i+1,k}$ is a simplicial retraction and $h_i\circ \pi_{a_i}$ and $\id_{T_{i,k}}$ are contiguous.

    Let $\sigma\in T_{i,k}$ be any simplex.
    So, $\sigma\in \cl\srips_{a_j}$ for some integer $j$ such that $i\le j\le k$.
    
    First, we prove that $\pi_{a_i}$ is a retraction.
    If $\sigma\in T_{i+1,k}$ then $j\ge i+1$.
    So, $\sigma\subseteq \cl\net_{a_j} \subseteq \net_{a_i}$ and thus $\pi_{a_i}(\sigma) = \sigma$ because $\pi_{a_i}$ is a retraction onto $\net_{a_i}$ by definition when viewed as a function on the vertex sets.
    
    Second, we show that $\pi_{a_i}$ is a simplicial map from $T_{i,k}$ to $T_{i+1,k}$.
    Since it is a retraction, it only remains to show that $\pi_{a_i}(\sigma)\in T_{i,k}$ when $\sigma\in T_{i,k}\setminus T_{i+1,k}$, i.e. when $j=i$.
    In this case, $\pi_{a_i}(\sigma) \in \srips_{a_i}$ because $\pi_{a_i}: \cl\srips_{a_i}\to\srips_{a_i}$ is simplicial (as shown in the proof of Lemma~\ref{lem:homology_isomorphic_sparsification}).
    Since $\srips_{a_i}\subseteq \cl\srips_{a_{i+1}}\subseteq T_{i,k}$, it follows that $\pi_{a_i}(\sigma)\in T_{i,k}$ as desired.
    
    Last, we prove contiguity.
    We need to prove that $\sigma \cup \pi_{a_i}(\sigma)\in T_{i,k}$.
    If $j > i$, then $\sigma \cup \pi_{a_i}(\sigma) = \sigma \in T_{i,k}$ as desired.
    If $i = j$, then $\sigma \cup \pi_{a_i}(\sigma)\in \cl\srips_{a_i}$ as shown in the proof of Lemma~\ref{lem:homology_isomorphic_sparsification}.
    Since $\cl\srips_{a_i} \subseteq T_{i,k}$, it follows that $\sigma \cup \pi_{a_i}(\sigma) \in T_{i,k}$ as desired.
  \end{proof}

  \begin{theorem}\label{thm:Dssrips_Dsrips}
    The persistence diagrams of $\srips$ and $\ssrips$ are identical.
  \end{theorem}
  \begin{proof}
    For any $a_i, a_{i+1} \in A$, we get the following commutative diagram where all maps are induced by inclusions.
    \[
      \xymatrix@C=0pt{
        {\Hom(\cl\srips_{a_i})} \ar[dr]_\cong && {\Hom(\srips_{a_i})} \ar[ll]_\cong \ar[dl] \ar[rr] && {\Hom(\cl\srips_{a_{i+1}})} \ar[dl]_\cong\\
        & {\Hom(\ssrips_{a_i})} \ar[rr]      && {\Hom(\ssrips_{a_{i+1}})} &
      }
    \]
        
    Lemma~\ref{lem:srips_in_ssrips_isomorphism} and Corollary~\ref{cor:back_arrows_are_isomorphisms} show that the indicated maps are isomorphisms.
    As in Section~\ref{sec:reversing_zigzags}, we reverse the isomorphism $\Hom(\srips_{a_i}) \to \Hom(\cl\srips_{a_i})$ to get the following diagrams, which also commutes.
    \[
      \xymatrix{
        \Hom(\cl\srips_{a_i}) \ar[d]_\cong \ar[r] & \Hom(\cl\srips_{a_{i+1}}) \ar[d]_\cong\\
        \Hom(\ssrips_{a_i}) \ar[r] & \Hom(\ssrips_{a_{i+1}}) 
      }
    \]
    Therefore, by the Persistence Equivalence Theorem, $\dgk\srips = \dgk\ssrips$.
  \end{proof}

    \subsection{The Connection with Extended Persistence} 
\label{sub:levelset_zigzag}

  The sparse Rips zigzag has the property that every other space is the intersection of its neighbors on either side.
  That is, by a simple exercise, one can show that $\srips_{a_i} = \cl\srips_{a_i}\cap\cl\srips_{a_{i+1}}$.
  Carlsson et al.\ give a general method for comparing such zigzags to filtrations that do not zigzag in their work on levelset zigzags induced by real-valued functions~\cite{carlsson09zigzag}.
  They proved that the Mayer-Vietoris diamond principle from a paper by Carlsson and De Silva~\cite{carlsson10zigzag} allows one to relate the persistence diagram of such a zigzag with the so-called extended persistence diagram of the union filtration.
  In our case, this result implies that the persistence diagram of the sparse zigzag Vietoris-Rips filtration can be derived from the persistence diagram of the extended filtration
  \[
    \xymatrix{
      {\Hom(T_{1,1})} \ar[r] &
      {\cdots} \ar[r] &
      {\Hom(T_{1,N})} \ar[d] \\
      {\Hom(T_{1,N}, T_{N,N})} &
      {\cdots} \ar[l] & 
      {\Hom(T_{N,N}, T_{1,N})} \ar[l] 
    }
  \]
  where $N = |A|$, $T_{i,k} := \bigcup_{j=i}^k \cl\srips_{a_j}$ as in the proof of Lemma~\ref{lem:srips_in_ssrips_isomorphism}, and $\Hom(T_{N,N}, T_{i,N})$ is the homology of $T_{N,N}$ relative to $T_{i,N}$, or equivalently, the homology of the quotient $T_{N,N}/ T_{i,N}$.
  The first half of this filtration is precisely the sparse Vietoris-Rips filtration $\ssrips$.
  
  In light of this result, we see that Lemma~\ref{lem:srips_in_ssrips_isomorphism} implies that there is nothing interesting happening in the second half of the extended filtration.
  In the language of extended persistence, this means that there are no extended or relative pairs.
  Thus, as shown in Theorem~\ref{thm:Dssrips_Dsrips}, there is no need to compute the extended persistence to compute the persistence of the sparse Vietoris-Rips filtration.


  \section{Theoretical Guarantees}\label{sec:theoretical_guarantees}
  There are two main theoretical guarantees regarding the sparse Vietoris--Rips filtrations.
  First, in Subsection~\ref{sec:approximation_guarantee}, we show that the resulting persistence diagrams are good approximations to the true Vietoris--Rips filtration.
  Second, in Subsection~\ref{sec:linear_complexity}, we show that the filtrations have linear size.

    \subsection{The Approximation Guarantee} 
\label{sec:approximation_guarantee}

  In this subsection we prove that the persistence diagram of the sparse Vietoris--Rips filtration is a multiplicative $c$-approximation to the persistence diagram of the standard Vietoris--Rips filtration, where $c = \frac{1}{1-2\e}$.
  The approach has two parts.
  First, we show that the relaxed filtration is a multiplicative $c$-approximation to the classical Vietoris--Rips filtration.
  Second, we show that the sparse and relaxed Vietoris--Rips filtrations have the same persistence diagrams, i.e.\ that $\dgk \srips = \dgk \rrips$.
  By passing through the filtration $\rrips$, we obviate the need to develop new stability results for zigzag persistence.
  
  \begin{theorem}\label{thm:approximation_guarantee}
    For any metric space $\M = (P,\dist)$, the persistence diagrams of the corresponding sparse Vietoris--Rips filtrations $\srips = \srips(\M)$ and $\ssrips = \ssrips(\M)$ both yield $c$-approximations to the persistence diagram of the Vietoris--Rips filtration $\rips=\rips(\M)$, where $c = \frac{1}{1-2\e}$ and $\e\le\frac{1}{3}$ is a user-defined constant.
  \end{theorem}
  \begin{proof}
    By Lemma~\ref{lem:rips_relaxed_rips_interleave}, we have a multiplicative $c$-interleaving between $\rips$ and $\rrips$.
    Thus, the Persistence Approximation Lemma implies that $\dgk \rrips$ is a $c$-approximation to $\dgk \rips$.

    We have shown in Theorem~\ref{thm:Dssrips_Dsrips} that $\dgk \srips = \dgk\ssrips$, so it will suffice to prove that $\dgk \srips = \dgk \rrips$.
    The rest of the proof follows the same pattern as in Theorem~\ref{thm:Dssrips_Dsrips}.
    For any $a_i, a_{i+1} \in A$, we get the following commutative diagram induced by inclusion maps.
    \[
      \xymatrix@C=0pt{
        {\Hom(\cl\srips_{a_i})} \ar[dr]_\cong && {\Hom(\srips_{a_i})} \ar[ll]_\cong \ar[dl]_\cong \ar[rr] && {\Hom(\cl\srips_{a_{i+1}})} \ar[dl]_\cong\\
        & {\Hom(\rrips_{a_i})} \ar[rr]      && {\Hom(\rrips_{a_{i+1}})} &
      }
    \]
    Corollary~\ref{cor:back_arrows_are_isomorphisms} implies that many of these inclusions induce isomorphisms at the homology level (as indicated in the diagram).
    As a consequence, the following diagram also commutes and the vertical maps are isomorphisms.
    \[
      \xymatrix{
        {\Hom(\cl\srips_{a_i})} \ar[d]_\cong \ar[r] & {\Hom(\cl\srips_{a_{i+1}})} \ar[d]_\cong\\
        {\Hom(\rrips_{a_i})} \ar[r]      & {\Hom(\rrips_{a_{i+1}})} 
      }
    \]
    So, the Persistence Equivalence Theorem implies that $\dgk \srips = \dgk \rrips$ as desired.
  \end{proof}

    \subsection{The Linear Complexity of the Sparse Filtration} 
\label{sec:linear_complexity}

  In this subsection, we prove that the total number of simplices in the sparse Vietoris--Rips filtration is only linear in the number of input points.
  We start by showing that the graph of all edges appearing in the filtration has only a linear number of edges.
  
  %

  For a point $p\in P$, let $E(p)$ be the set of neighbors of $p$ whose removal time is at least as large as that of $p$:
  \[
    E(p) := \{q\in P : t_p \le t_q \text{ and } (p,q)\in \cl\srips_{t_p}\}.
  \]
  To compute the filtrations $\srips$ and $\ssrips$, it suffices to compute $E(p)$ for each $p\in P$.
  In fact $\ssrips_\infty$ is just the clique complex on the graph of all edges $(p,q)$ such that $q\in E(p)$.
  
  \begin{lemma}\label{lem:linear_complexity}
    Given a set of $n$ points in a metric space $\M = (P, \dist)$ with doubling dimension at most $d$ and a net-tree with parameter $\e\le \frac{1}{3}$, the cardinality $|E(p)|$ is at most $\frac{1}{\e}^{O(d)}$ for each $p\in P$.
  \end{lemma}
  \begin{proof}
    Let $\spread(E(p))$ denote the spread of $E(p)$.
    Since $E(p)$ is a finite metric with doubling dimension at most $d$, the number of points is at most $\spread(E(p))^{O(d)}$.
    So, it will suffice to prove that for all $p\in P$, $\spread(E(p)) = O(\frac{1}{\e})$.

    The definition of $E(p)$ implies that $E(p)\subseteq \net_{t_p}$ and so by Lemma~\ref{lem:nets_pack_and_cover},
    the nearest pair in $E(p)$ are at least $\kpack\e(1-2\e)t_p$ apart.
    For $q\in E(p)$, since $(p,q)\in \srips_{t_p}$, $\dist(p,q) \le \rdist_{t_p}(p,q) \le t_p$.
    It follows that
    the farthest pair in $E(p)$ are at most $2t_p$ apart.
    So, we get that $\spread(E(p)) \le \frac{2t_p}{\kpack\e(1-2\e)t_p} = O(\frac{1}{\e})$ as desired.
  \end{proof}
  
  We see that the size of the graph in the filtration is governed by three variables: the doubling dimension, $d$; the packing constant of the net-tree, $\kpack$; and the desired tightness of the approximation, $\e$.
  The preceding Lemma easily implies the following bound on the higher order simplices.

  \begin{theorem}\label{thm:full_linear_complexity}
    Given a set of $n$ points in a metric space $\M = (P, \dist)$ with doubling dimension $d$, the total number of $k$-simplices in the sparse Vietoris--Rips filtrations $\srips$ and $\ssrips$ is at most $\left(\frac{1}{\e}\right)^{O(kd)}n$.
  \end{theorem}

    
  \section{An algorithm to construct the sparse filtration} 
\label{sec:construction}

  The net-tree defines the deletion times of the input points and thus determines the perturbed metric.
  It also gives the necessary data structure to efficiently find the neighbors of a point in the perturbed metric in order to compute the filtration.
  In fact, this is exactly the kind of search that the net-tree makes easy.
  Then we find all cliques, which takes linear time because each is subset of $E(p)$ for some $p\in P$ and each $E(p)$ has constant size.

  As explained in the Har-Peled and Mendel paper~\cite{har-peled06fast}, it is often useful to augment the net-tree with ``cross'' edges connecting nodes at the same level in the tree that are represented by geometrically close points.
  The set of \emph{relatives} of a node $u\in T$ is defined as
  \begin{align*}
    \rel(u) := \{ v\in T : &\rad(v)\le\rad(u)<\rad(\parent(v)) \text{ and } \\
          & \dist(\rep(u), \rep(v)) \le C \rad(u) \},
  \end{align*}
  where $C$ is a constant bigger than $3$.%
  \footnote{The precise value of $C$ depends on some constants chosen in the construction of the net-tree and can be extracted from the Har-Peled and Mendel paper.  For our purposes, we only need the fact that it is bigger than $3$.}
  The size of $\rel(u)$ is a constant using the same packing arguments as in Lemma~\ref{lem:linear_complexity}.

  \begin{figure}[ht]
    \centering
      \includegraphics[width=\textwidth]{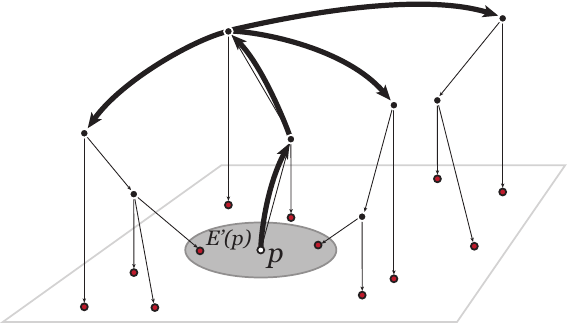}
    \caption{\label{fig:net_tree_search}
      To find the nearby points at or above the level of $v_p$ in the net-tree only requires traveling up a constant number of levels and then searching the relative trees.
    }
  \end{figure}

  This makes it easy to do a range search to find the points of $E(p)$.
  In fact, we will find the slightly larger set $E'(p) = \cl\net_{t_p}\cap \ball(p, t_p)$.
  The search starts by finding $u$ the highest ancestor of $v_p$ whose radius is at most some fixed constant times $t_p$.
  Since the radius increases by a constant factor on each level, this is only a constant number of levels.
  Then the subtrees rooted at each $v\in \rel(u)$ are searched down to the level of $v_p$.
  Thus, we search a constant number of trees of constant degree down a constant number of levels.
  The resulting search finds all of the points of $E(p)$ in constant time.

  Since the work is only constant time per point, the only superlinear work is in the computation of the net-tree.
  As noted before, this requires only $O(n\log n)$ time.
  
  \section{Conclusions and Directions for Future Work} 
\label{sec:conclusion}

  We have presented an efficient method for approximating the persistent homology of the Vietoris--Rips filtration.
  Computing these approximate persistence diagrams at all scales has the potential to make persistence-based methods on metric spaces tractable for much larger inputs.

  Adapting the proofs given in this paper to the \v Cech filtration is a simple exercise.
  Moreover, it may be possible to apply a similar sparsification to complexes filtered by alternative distance-like functions like the distance to a measure introduced by Chazal et al.~\cite{chazal11geometric}.

  Another direction for future work is to identify a more general class of hierarchical structures that may be used in such a construction.
  The net-tree used in this paper is just one example chosen primarily because it can be computed efficiently.

  The analytic technique used in this paper may find more uses in the future.
  We effectively bounded the difference between the persistence diagrams of a filtration and a zigzag filtration by embedding the zigzag filtration in a topologically equivalent filtration that does not zigzag at the homology level.
  This is very similar to the relationship between the levelset zigzag and extended persistence demonstrated by Carlsson et al.~\cite{carlsson09zigzag}.
  In that paper, such a technique gave some stability results for levelset zigzags of real-valued Morse functions on manifolds.
  It may be that other zigzag filtrations can be analyzed in this way.


  
  \section*{Acknowledgements} 
\label{sec:acknowledgements}

I would like to thank Steve Oudot and Fr\'{e}d\'{e}ric Chazal for their helpful comments.
I would also like to acknowledge the keen insight of David Cohen-Steiner who suggested that I look for a non-zigzagging filtration.
This insight became the filtration $\ssrips$ of Section~\ref{sec:no_zigzag}.
This work was partially supported by GIGA grant ANR-09-BLAN-0331-01 and the European project CG-Learning No. 255827.


  \clearpage
  
  \bibliographystyle{plain}
  \bibliography{sparse_rips}

\begin{thebibliography}{10}

\bibitem{attali11efficient}
Dominique Attali, Andr{\'e} Lieutier, and David Salinas.
\newblock Efficient data structure for representing and simplifying simplicial
  complexes in high dimension.
\newblock In {\em Proceedings of the 27th {ACM} Symposium on Computational
  Geometry}, 2011.

\bibitem{boissonat07manifold}
Jean-Daniel Boissonat, Leonidas~J. Guibas, and Steve~Y. Oudot.
\newblock Manifold reconstruction in arbitrary dimensions using witness
  complexes.
\newblock In {\em Proceedings of the 23rd {ACM} Symposium on Computational
  Geometry}, 2007.

\bibitem{carlsson09topology}
Gunnar Carlsson.
\newblock Topology and data.
\newblock {\em Bull. Amer. Math. Soc.}, 46:255--308, 2009.

\bibitem{carlsson10zigzag}
Gunnar Carlsson and Vin de~Silva.
\newblock Zigzag persistence.
\newblock {\em Foundations of Computational Mathematics}, 10(4):367--405, 2010.

\bibitem{carlsson09zigzag}
Gunnar Carlsson, Vin de~Silva, and Dmitriy Morozov.
\newblock Zigzag persistent homology and real-valued functions.
\newblock In {\em Proceedings of the 25th {ACM} Symposium on Computational
  Geometry}, pages 247--256, 2009.

\bibitem{carlsson08local}
Gunnar Carlsson, Tigran Ishkhanov, Vin de~Silva, and Afra Zomorodian.
\newblock On the local behavior of spaces of natural images.
\newblock {\em International Journal of Computer Vision}, 76(1):1--12, 2008.

\bibitem{chazal09proximity}
Fr{\'e}d{\'e}ric Chazal, David Cohen-Steiner, Marc Glisse, Leonidas~J. Guibas,
  and Steve~Y. Oudot.
\newblock Proximity of persistence modules and their diagrams.
\newblock In {\em Proceedings of the 25th {ACM} Symposium on Computational
  Geometry}, pages 237--246, 2009.

\bibitem{chazal09gromov-hausdorff}
Fr{\'e}d{\'e}ric Chazal, David Cohen-Steiner, Leonidas~J. Guibas, Facundo
  M\'{e}moli, and Steve~Y. Oudot.
\newblock Gromov-hausdorff stable signatures for shapes using persistence.
\newblock In {\em Eurographics Symposium on Geometry Processing}, 2009.

\bibitem{chazal11geometric}
Fr{\'e}d{\'e}ric Chazal, David Cohen-Steiner, and Quentin M{\'e}rigot.
\newblock Geometric inference for probability measures.
\newblock {\em Foundations of Computational Mathematics}, 11:733--751, 2011.

\bibitem{chazal08towards}
Fr{\'e}d{\'{e}}ric Chazal and Steve~Y. Oudot.
\newblock Towards persistence-based reconstruction in {E}uclidean spaces.
\newblock In {\em Proceedings of the 24th {ACM} Symposium on Computational
  Geometry}, pages 232--241, 2008.

\bibitem{chung09persistence}
Moo~K. Chung, Peter Bubenik, and Peter~T. Kim.
\newblock Persistence diagrams of cortical surface data.
\newblock In {\em Information Processing in Medical Imaging}, volume 5636,
  pages 386--397, 2009.

\bibitem{clarkson99nearest}
Kenneth~L. Clarkson.
\newblock Nearest neighbor queries in metric spaces.
\newblock {\em Discrete \& Computational Geometry}, 22(1), 1999.

\bibitem{clarkson03nearest}
Kenneth~L. Clarkson.
\newblock Nearest neighbor searching in metric spaces: Experimental results for
  {$\mathrm{sb}(S)$}.
\newblock Preliminary version presented at ALENEX99, 2003.

\bibitem{clarkson06building}
Kenneth~L. Clarkson.
\newblock Building triangulations using epsilon-nets.
\newblock In {\em STOC: ACM Symposium on Theory of Computing}, pages 326--335,
  2006.

\bibitem{cohen-steiner07stability}
David Cohen-Steiner, Herbert Edelsbrunner, and John Harer.
\newblock Stability of persistence diagrams.
\newblock {\em Discrete \& Computational Geometry}, 37(1):103--120, 2007.

\bibitem{cole06searching}
Richard Cole and Lee-Ad Gottlieb.
\newblock Searching dynamic point sets in spaces with bounded doubling
  dimension.
\newblock In {\em Proceedings of the 38th Annual ACM Symposium on Theory of
  Computing}, pages 574--583, 2006.

\bibitem{de-silva04topological}
V.~de~Silva and G.~Carlsson.
\newblock Topological estimation using witness complexes.
\newblock In {\em Proc. Sympos. Point-Based Graphics}, pages 157--166, 2004.

\bibitem{silva07coverage}
Vin de~Silva and Robert Ghrist.
\newblock Coverage in sensor networks via persistent homology.
\newblock {\em Algorithmic \& Geometric Topology}, 7:339--358, 2007.

\bibitem{silva07homological}
Vin de~Silva and Robert Ghrist.
\newblock Homological sensor networks.
\newblock {\em Notices Amer. Math. Soc.}, 54(1):10--17, 2007.

\bibitem{edelsbrunner09computational}
Herbert Edelsbrunner and John~L. Harer.
\newblock {\em Computational Topology: An Introduction}.
\newblock Amer. Math. Soc., 2009.

\bibitem{edelsbrunner02topological}
Herbert Edelsbrunner, David Letscher, and Afra Zomorodian.
\newblock Topological persistence and simplification.
\newblock {\em Discrete \& Computational Geometry}, 4(28):511--533, 2002.

\bibitem{gao06deformable}
Jie Gao, Leonidas~J. Guibas, and An~Thai Nguyen.
\newblock Deformable spanners and applications.
\newblock {\em Comput. Geom}, 35(1-2):2--19, 2006.

\bibitem{gottlieb08optimal}
Lee-Ad Gottlieb and Liam Roditty.
\newblock An optimal dynamic spanner for doubling metric spaces.
\newblock In {\em Proceedings of the 16th annual European symposium on
  Algorithms}, pages 468--489, 2008.

\bibitem{guibas07reconstruction}
Leonidas~J. Guibas and Steve~Y. Oudot.
\newblock Reconstruction using witness complexes.
\newblock In {\em Proceedings 18th ACM-SIAM Symposium: Discrete Algorithms},
  pages 1076--1085, 2007.

\bibitem{har-peled06fast}
Sariel Har-Peled and Manor Mendel.
\newblock Fast construction of nets in low dimensional metrics, and their
  applications.
\newblock {\em SIAM Journal on Computing}, 35(5):1148--1184, 2006.

\bibitem{hudson10topological}
Beno\^{\i}t Hudson, Gary~L. Miller, Steve~Y. Oudot, and Donald~R. Sheehy.
\newblock Topological inference via meshing.
\newblock In {\em Proceedings of the 26th {ACM} Symposium on Computational
  Geometry}, pages 277--286, 2010.

\bibitem{milosavljevic11zigzag}
Nikola Milosavljevic, Dmitriy Morozov, and Primoz Skraba.
\newblock Zigzag persistent homology in matrix multiplication time.
\newblock In {\em Proceedings of the 27th {ACM} Symposium on Computational
  Geometry}, 2011.

\bibitem{munkres84elements}
James~R. Munkres.
\newblock {\em Elements of Algebraic Topology}.
\newblock Addison-Wesley, 1984.

\bibitem{narasimhan07geometric}
Giri Narasimhan and Michiel H.~M. Smid.
\newblock {\em Geometric Spanner Networks}.
\newblock Cambridge University Press, 2007.

\bibitem{singh08topological}
Gurjeet Singh, Facundo M\'{e}moli, Tigran Ishkhanov, Guillermo Sapiro, Gunnar
  Carlsson, and Dario~L. Ringach.
\newblock Topological analysis of population activity in visual cortex.
\newblock {\em Journal of Vision}, 8(8):1--18, 2008.

\bibitem{tausz11applications}
Andrew Tausz and Gunnar Carlsson.
\newblock Applications of zigzag persistence to topological data analysis.
\newblock arxiv:1108.3545, Aug 2011.

\bibitem{zomorodian10tidy}
Afra Zomorodian.
\newblock The tidy set.
\newblock In {\em Proceedings of the 26th {ACM} Symposium on Computational
  Geometry}, pages 257--266, 2010.

\bibitem{zomorodian05computing}
Afra Zomorodian and Gunnar Carlsson.
\newblock Computing persistent homology.
\newblock {\em Discrete \& Computational Geometry}, 33(2):249--274, 2005.

\end{thebibliography}
  
\end{document}